\documentclass[11pt]{article}

\usepackage{fullpage}
\usepackage{times}
\usepackage{latexsym}
\usepackage{amsfonts,amssymb,amsthm}
\usepackage{euscript}
\usepackage{amstext}
\usepackage{algorithm,algorithmic}
\usepackage{graphicx}
\usepackage{color}
\usepackage{url,hyperref}
\usepackage{enumerate}
\usepackage{xspace}
\usepackage{dsfont}
\usepackage{paralist}

\newcommand{\eps}{\epsilon}
\newcommand{\etal}{et al.\ }
\newcommand{\opt}{\textsc{OPT}}
\renewcommand{\Pr}{\textbf{Pr}}
\newcommand{\Ex}{\textbf{E}}

\newtheorem{lemma}{Lemma}[section]
\newtheorem{theorem}[lemma]{Theorem}

\newtheorem{corollary}[lemma]{Corollary}
\newtheorem{proposition}[lemma]{Proposition}

\newcommand{\mrc}{\mathcal{MRC}}
\newcommand{\kv}{\ensuremath{\langle key; value \rangle}}
\newcommand{\mapkc}{\texttt{MapReduce-kCenter}\xspace}
\newcommand{\mapkm}{\texttt{MapReduce-kMedian}\xspace}
\newcommand{\iters}{\texttt{Iterative-Sample}\xspace}
\newcommand{\itersmr}{\texttt{MapReduce-Iterative-Sample}\xspace}
\newcommand{\select}{\texttt{Select}\xspace}
\newcommand{\dividekm}{\texttt{MapReduce-Divide-kMedian}\xspace}
\newcommand{\pv}{v}
\newcommand{\rank}{\texttt{rank}}

\def\script#1{\mathcal{#1}}

\begin{document}

\title{Fast Clustering using MapReduce}

\author{
Alina Ene\thanks{
University of Illinois, Urbana, IL, USA. ene1@illinois.edu. Partially supported by NSF grants CCF-0728782 and CCF-1016684.}
\and
Sungjin Im\thanks{ 
University of Illinois,
Urbana, IL, USA,
im3@illinois.edu. Partially supported by NSF grants
CCF-0728782, CCF-1016684, and a Samsung Fellowship.}
\and
Benjamin Moseley\thanks{  
University of Illinois,
Urbana, IL, USA,
bmosele2@illinois.edu.Partially supported by NSF grants
CCF-0728782 and CCF-1016684.}
}

\maketitle

\begin{abstract}
Clustering problems have numerous applications and are becoming more
challenging as the size of the data increases.  In this paper, we
consider designing clustering algorithms that can be used in
MapReduce, the most popular programming environment for processing
large datasets. We focus on the practical and popular clustering
problems, $k$-center and $k$-median. We develop fast clustering
algorithms with constant factor approximation guarantees. From a
theoretical perspective, we give the first analysis that shows
several clustering algorithms are in $\mathcal{MRC}^0$, a theoretical
MapReduce class introduced by Karloff et al. \cite{KarloffSV10}. Our
algorithms use sampling to decrease the data size and they run a time
consuming clustering algorithm such as local search or Lloyd's
algorithm on the resulting data set. Our algorithms have sufficient
flexibility to be used in practice since they run in a constant
number of MapReduce rounds. We complement these results by performing
experiments using our algorithms.  We compare the empirical
performance of our algorithms to several sequential and parallel
algorithms for the  $k$-median problem. The experiments show that our
algorithms' solutions are similar to or better than the other
algorithms' solutions. Furthermore, on data sets that are
sufficiently large, our algorithms are faster than the other parallel
algorithms that we tested.
\end{abstract}

\section{Introduction}

Clustering data is a fundamental problem in a variety of areas of
computer science and related fields.  Machine learning, data mining,
pattern recognition, networking, and bioinformatics use clustering
for data analysis. Consequently, there is a vast amount of research
focused on the topic
\cite{ArthurV07,MaS07,Chen08,Chen06,AryaGKMMP04,Thorup04,CharikarCGG98,DyerF85,CharikarCFM04,McCutchenK08}. In the clustering problems that we consider in this paper,
the goal is to partition the data into subsets, called clusters,
such that the data points assigned to the same cluster are similar
according to some metric. 

In several applications, it is of interest to classify or group web pages
according to their content or cluster users based on their online behavior. One
such example is finding communities in social networks. Communities consist of
individuals that are closely related according to some relationship criteria.
Finding these communities is of interest for applications such as predicting
buying behavior or designing targeted marking plans and is an ideal application
for clustering. However, the size of the web graph and social network graphs
can be quite large; for instance, the web graph consists of a trillion edges
\cite{MalewiczABDHLC10}. When the amount of data is this large, it is difficult
or even impossible for the data to be stored on a single machine, which renders
sequential algorithms unusable. In situations where the amount of data is
prohibitively large, the MapReduce \cite{dean-mapreduce} programming paradigm
is used to overcome this obstacle.  MapReduce and its open source counterpart
Hadoop \cite{white-hadoop} are distributed computing frameworks designed to
process massive data sets.

The MapReduce model is quite novel, since it interleaves sequential and
parallel computation. Succinctly, MapReduce consists of several \emph{rounds}
of computation.  There is a set of machines, each of which has a certain amount
of memory available. The memory on each machine is limited, and there is no
communication between the machines during a round.  In each round, the data is
distributed among the machines.  The data assigned to a single machine is
constrained to be sub-linear in the input size.  This restriction is motivated
by the fact that the input size is assumed to be very large
\cite{KarloffSV10,ChierichettiKT10}.  After the data is distributed, each of
the machines performs some computation on the data that is available to them.
The output of these computations is either the final result or it becomes the
input of another MapReduce round. A more precise overview of the MapReduce
model is given in Section~\ref{sec:mr}.

\smallskip\noindent
\textbf{Problems:}
In this paper, we are concerned with designing clustering algorithms
that can be implemented using MapReduce. In particular, we focus on
two well-studied problems: metric $k$-median and $k$-center.  In both
of these problems, we are given a set $V$ of $n$ points, together with
the distances between any pair of points; we give a precise
description of the input representation below. The goal is to choose
$k$ of the points. Each of the $k$ chosen points represents a cluster
and is referred to as a \emph{center}.  Every data point is assigned
to the closest center and all of the points assigned to a given point
form a cluster.  In the $k$-center problem, the goal is to choose the
centers such that the maximum distance between a center and a point
assigned to it is minimized. In the $k$-median problem the objective
is to minimize the sum of the distances from the centers to each of
the points assigned to the centers. Both of the problems are known to
be NP-Hard. 
Thus previous work has focused on finding approximation algorithms
\cite{HochbaumS85,Bartal96,CharikarGTS02,
CharikarG99,AryaGKMMP04,GuptaT08,BlellochT10}. Many of the existing algorithms
are inherently sequential and, with the exception of the algorithms of
\cite{BlellochT10,GuhaMMMO03}, they are difficult to adapt to a parallel
computing setting. We discuss the algorithms of \cite{GuhaMMMO03} in more detail
later.

\smallskip\noindent
\textbf{Input Representation:}
Let $d: V \times V \rightarrow \mathbb{R}_+$ denote the distance function. The
distance function $d$ is a metric, i.e., it has the following properties: (1)
$d(x, y) = 0$ if and only if $x = y$, (2) $d(x, y) = d(y, x)$ for all $x, y$,
and (3) $d(x, z) \leq d(x, y) + d(y, z)$ for all $x, y, z$.  The third property
is called the triangle inequality; we note that our algorithms only rely on the
fact that the distances between points satisfy the triangle inequality.

Now we discuss how the distance function is given to our algorithms.  In some
settings, the distance function has an implicit compact representation; for
example, if the distances between points are shortest path distances in a
sparse graph, the graph represents the distance function compactly. However,
currently there does not exist a MapReduce algorithm that computes shortest
paths in a constant number of rounds, even if the graph is unweighted. This
motivates the assumption that we are given the distance function
\emph{explicitly} as a set of $\Theta(n^2)$ distances, one for each pair of
points, or we are given access to an \emph{oracle} that takes as input two
points and returns the distance between them.  Throughout this paper, we assume
that the distance function is given explicitly. More precisely, we assume that
the input is a \emph{weighted complete graph} $G = (V, E)$ that has an edge
$xy$ between any two points in $V$, and the weight of the edge $xy$ is $d(x,
y)$\footnote{We note that some of the techniques in this paper extend to the
setting in which the distance function is given as an oracle.}. Moreover, we
assume that $k$ is at most $O(n^{1 - \delta})$ for some constant $\delta > 0$,
and the distance between any pair of points is upper bounded by some polynomial
in $n$. These assumptions are justified in part by the fact that the number of
points is very large, and by the memory constraints of the MapReduce model; we
discuss the MapReduce model in more detail in Section~\ref{sec:mr}.

\smallskip\noindent
\textbf{Contributions}: We introduce the \emph{first} approximate
metric $k$-median and $k$-center algorithms designed to run on
MapReduce. More precisely, we show the following results.

\begin{theorem}\label{thm:kcenter}
    There is a randomized constant approximation algorithm for the $k$-center
    problem that, with high probability,  runs in  $O(\frac{1}{\delta})$
     MapReduce rounds and uses memory at most $O(k^2 n^{\delta})$ on
    each of the machines for any constant $\delta > 0$.
\end{theorem}

\begin{theorem}\label{thm:kmedian}
    There is a randomized constant approximation algorithm  for the $k$-median
    problem that, with high probability, runs in $O(\frac{1}{\delta})$ MapReduce rounds and uses memory at most $O(k^2 n^{\delta})$ on
    each of the machines for any constant $\delta > 0$.
\end{theorem}


To complement these results, we run our algorithms on randomly
generated data sets.   For the $k$-median problem we compare our
algorithm to a parallelized implementation of Lloyd's algorithm
\cite{Lloyd82,blog,googlelecture}, arguably the most popular
clustering algorithm used in practice (see \cite{AgarwalM04,Herwig99}
for example), the local search algorithm \cite{AryaGKMMP04,GuptaT08},
the best known approximation algorithm for the $k$-median problem and
a partitioning based algorithm that can parallelize any sequential
clustering algorithm (see Section \ref{sec:exp}). Our
algorithms achieve a speed-up of 1000x  over the local search
algorithm and 20x over the parallelized Lloyd's algorithm, a
significant improvement in running time.  Further, our algorithm's
objective is similar to Lloyd's algorithm and the local search
algorithm.  For the partitioning based algorithm, we show that our
algorithm achieves faster running time when the number of points is
large.  Thus for the $k$-median problem our algorithms are fast with a
small loss in performance.  For the $k$-center problem we compare our
algorithm to the well known algorithm of \cite{DyerF85,Gonzalez85},
which is the best approximation algorithm for the problem and is quite
efficient.  Unfortunately, for the $k$-center problem our algorithm's
objective is a factor four worse in some cases. This is due to the
sensitivity of the $k$-center objective to sampling.

Our algorithms show that the $k$-center and $k$-median problem belong
to the theoretical MapReduce class $\mrc^0$ that was introduced by
Karloff \etal \cite{KarloffSV10}\footnote{Recall that we only consider
instances of the problems in which $k$ is sub-linear in the number of
points, and the distances between points are upper bounded by some
polynomial in $n$.}. Let $N$ denote the total size of the input, and
let $0< \eps < 1$ be a fixed constant. A problem is in the
MapReduce class $\mrc^0$ if it can be solved using a constant number
of rounds and an $O(N^{1 - \eps})$ number of machines, where each
machine has $O(N^{1 - \eps})$ memory available \cite{KarloffSV10}.
Differently said, the problem has an algorithm that uses a sub-linear
amount of memory on each machine and a sub-linear number of machines.
One of the main motivations for these restrictions is that a typical
MapReduce input is very large and it might not be possible to store
the entire input on a single machine.  Moreover, the size of the input
might be much larger than the number of machines available. We discuss
the theoretical MapReduce model in Section~\ref{sec:mr}.  Our assumptions on the size of $k$ and the point distances are needed in order to show that the memory that our
algorithms use on each machine is sub-linear in the total input size.
For instance, without the assumption on $k$, we will not be able to fit $k$ points in
the memory available on a machine.


\smallskip\noindent
\textbf{Adapting Existing Algorithms to MapReduce:}
Previous work on designing algorithms for MapReduce are generally based on the
following approach. Partition the input and assign each partition to a unique machine.  On each
machine, we perform some computation that eliminates a large fraction of the
input. We collect the results of this computations on a single machine, which
can store the data since the data has been sparsified. On this machine, we
perform some computation and we return the final solution. We can use a similar
approach for the $k$-center and $k$-median problems. We partition the points
across the machines.  We cluster each of the partitions.  We select one point
from each cluster and put all of the selected points on a single machine. We
cluster these points and output the solution. Indeed, a similar algorithm was
considered by Guha \etal \cite{GuhaMMMO03} for the $k$-median problem in the
streaming model.  We give the details of how to implement this algorithm in
MapReduce in Section~\ref{sec:exp} along with an analysis of the algorithm's
approximation guarantees. Unfortunately, the total running time for the
algorithm can be quite large, since it runs a costly clustering algorithm on
$\Omega(k\sqrt{n/k})$ points.  Further, this algorithm requires $\Omega(k n)$
memory on each of the machines.

Another strategy for developing algorithms for $k$-center and
$k$-median that run in MapReduce is to try to adapt existing parallel
algorithms.  To the best of our knowledge, the only parallel
algorithms known with provable guarantees were given by Blelloch and
Tangwongsan \cite{BlellochT10}; Blelloch and Tangwongsan
\cite{BlellochT10} give the first PRAM algorithms for $k$-center and
$k$-median. Unfortunately, these algorithms assume that the number of
machines available is $\Omega(N^2)$, where $N$ is the total input
size, and there is some memory available in the system that can be
accessed by all of the machines. These assumptions are too strong for
the algorithms to be in used in MapReduce. Indeed, the requirements
that the machines have a limited amount of memory and that there is no
communication between the machines is what differentiates the
MapReduce model from standard parallel computing models. Another
approach is to try to adapt algorithms that were designed for the
streaming model. Guha \etal \cite{GuhaMMMO03} have given a $k$-median
algorithm for the streaming model; with some work, we can adapt one of
the algorithms in \cite{GuhaMMMO03} to the MapReduce model.
However, this algorithm's approximation ratio degrades exponentially
in the number of rounds.

\smallskip\noindent
\textbf{Related Work}:  There has been a large amount of work on the
metric $k$-median and $k$-center problems. Due to space constraints,
we focus only on closely related work that we have not already
mentioned. Both problems are known to be NP-Hard.
Bartal \cite{Bartal98,Bartal96} gave an algorithm for the
$k$-median problem that achieves an $O(\log n \log \log n)$ approximation
ratio.
Later Charikar \etal gave the first constant factor approximation of
$6+\frac{2}{3}$ \cite{CharikarGTS02}.  This approach was based on LP
rounding techniques.  The best known approximation algorithm achieves
a $3 + \frac{2}{c}$ approximation in $O(n^c)$ time
\cite{AryaGKMMP04, GuptaT08}; this algorithm is based on the local
search technique. On the other hand, Jain \etal \cite{JainMS02} have
shown that there does not exist an $1 +  (2 / e)$ approximation for
the $k$-median problem unless $\textsc{NP} \subseteq
\textsc{DTIME}(n^{O(\log \log n)})$.  For the $k$-center problem, two
simple algorithms are known which achieve a $2$-approximation
\cite{HochbaumS85,DyerF85,Gonzalez85} and this approximation ratio is
tight assuming that P $\neq$ NP.

MapReduce has received a significant amount of attention recently.
Most previous work has been on designing practical heuristics to solve
large scale problems \cite{KangTAFL08,LinD10}. Recent papers
\cite{KarloffSV10, FeldmanMSSS10} have focused on developing
computational models that abstract the power and limitations of
MapReduce.  Finally, there has been work on developing algorithms and
approximation algorithms that fit into the MapReduce model
\cite{KarloffSV10,ChierichettiKT10}.  This line of work has shown that
problems such as minimum spanning tree,  maximum coverage, and
connectivity can be solved efficiently using MapReduce.

\subsection{MapReduce Overview}
\label{sec:mr}
In this section we give a high-level overview of the MapReduce
model; for a more detailed description, see \cite{KarloffSV10}. The
data is represented as $\kv$ pairs. The \ensuremath{key} acts as an
address of the machine to which the \ensuremath{value} needs to be
sent to. A MapReduce round consists of three stages: map, shuffle,
and reduce. The map phase processes the data as follows. The
algorithm designer specifies a map function $\mu$, which we refer to
as a \emph{mapper}. The mapper takes as input a $\kv$ pair, and it
outputs a sequence of $\kv$ pairs. Intuitively, the mapper maps the
data stored in the $\kv$ pair to a machine. In the map phase, the
map function is applied to all $\kv$ pairs. In the shuffle phase,
all $\kv$ pairs with a given key are sent to the same machine; this
is done automatically by the underlying system. The reduce phase
processes the $\kv$ pairs created in the map phase as follows. The
algorithm designer specifies a reduce function $\rho$, which we
refer to as a \emph{reducer}. The reducer takes as input all the
$\kv$ pairs that have the same key, and it outputs a sequence of
$\kv$ pairs which have the same key as the input pairs; these pairs
are either the final output, or they become the input of the next
MapReduce round. Intuitively, the reducer performs some sequential
computation on all the data that is stored on a machine. The mappers
and reducers are constrained to run in time that is polynomial in
the size of the initial input, and not their input.


The theoretical $\mrc$ class was introduced in \cite{KarloffSV10}. The
class is designed to capture the practical restrictions of MapReduce
as faithfully as possible; a detailed justification of the model can
be found in \cite{KarloffSV10}. In addition to the constraints on the
mappers and reducers, there are three types of restrictions in $\mrc$:
constraints on the number of machines used, on the memory available on
each of the machines, and on the number of rounds of computation.  If
the input to a problem is of size $N$ then an algorithm is in $\mrc$
if it uses at most $N^{1-\epsilon}$ machines, each with at most
$N^{1-\epsilon}$ memory for some constant $\epsilon >0$\footnote{The
algorithm designer can choose $\epsilon$.}.  Notice that this implies
that the total memory available  is $O(N^{2-2\eps})$.  Thus the
difficulty of designing algorithms for the MapReduce model does not
come from the lack of total memory. Rather, it stems from the fact
that the memory available on each machine is limited; in particular,
the entire input does not fit on a single machine.  Not allowing the
entire input to be placed on a single machine makes designing
algorithms difficult, since a machine is only aware of a subset of the
input.  Indeed, because of this restriction, it is currently not known
whether fundamental graph problems such as shortest paths or maximum
matchings can be computed in a constant number of rounds, even if the
graphs are unweighted.

In the following, we state the precise restrictions on the resources
available to an algorithm for a problem in the class $\mrc^0$.

\smallskip
\begin{compactitem}
\item \textbf{Memory}:  The total memory used on a specific machine is
at most $O(N^{1-\eps})$.
\item  \textbf{Machines}: The total number of machines used is
$O(N^{1-\eps})$.
\item \textbf{Rounds}: The number of rounds is constant.
\end{compactitem}


\section{Algorithms}
In this section we describe our clustering algorithms \mapkc and
\mapkm. For both of our algorithms, we will parameterize the amount of
memory needed on a machine. For the MapReduce setting, the amount of
memory our algorithms require on each of the machines is parameterized
by $\delta > 0$ and we assume that the memory  is
$\Omega(k^2n^\delta)$.  It is further assumed that the number of
machines is large enough to store all of the input data across the
machines. Both algorithms use \iters as a subroutine which uses
sampling ideas from \cite{Thorup04}. The role of \iters is to get a
substantially smaller subset of points that represents all of the
points well. To achieve this, \iters performs the following
computation iteratively: in each iteration, it adds a small sample of
points to the final sample, it determines which points are ``well
represented'' by the sample, and it recursively considers only the
points that are not well represented.  More precisely, after sampling,
$\iters$ discards most points that are close to the current sample,
and it recurses on the remaining (unsampled) points. The algorithm
repeats this procedure until the number of points that are still
unrepresented is small and all such points are added to the sample.
Once we have a good sample, we run a clustering algorithm on just the
sampled points.  Knowing that the sampling represents all
unsampled points well, a good clustering of the sampled points will also be
a good clustering of all of the points.  Here the clustering algorithm
used will depend on the problem considered.  In the following section,
we show how $\iters$ can be implemented in the sequential setting to
highlight the high level ideas.  Then we show how to extend the
algorithms to the MapReduce setting.

\subsection{Sampling Sequentially}
In this section, the sequential version of the sampling algorithm is discussed.
When we mention the distance of a point $x$ to a set $S$, we mean the minimum
distance between $x$ and any point in $S$.  Our algorithm is parameterized by a
constant $0<\eps < \frac{\delta}{2}$ whose value can be changed depending on
the system specifications.  Simply, the value of $\eps$ determines the sample
size. For each of our algorithms there is a natural
trade-off between the sample size and the running time of the
algorithm.

\begin{algorithm}[h!] \caption{$\iters(V, E, k, \eps)$:} \label{alg:iters}
    \begin{algorithmic}[1]

    \STATE Set $S \leftarrow \emptyset$, $R \leftarrow V$.

    \WHILE{ $|R| > \frac{4}{\eps} k n^{\eps} \log n$}

    \STATE Add each point in $R$ with probability $\frac{9k
    n^\eps}{|R|} \log n$ independently to $S$.

    \STATE Add each point in $R$ with probability $\frac{4
    n^\eps}{|R|} \log n $ independently to $H$.

    \STATE $\pv \leftarrow \select(H,S)$

    \STATE Find the distance of each point $x \in R$ to $S$.  Remove $x$ from $R$ if this distance is smaller than the distance of $\pv$ to $S$.

    \ENDWHILE

    \STATE Output $ C := S \cup R$

    \end{algorithmic}
\end{algorithm}

\begin{algorithm}[h!] \caption{$\select(H,S)$:} \label{alg:select}
    \begin{algorithmic}[1]

    \STATE For each point $x \in H$, find the distance of $x$ to $S$.

    \STATE Order the points in $H$ according to their distance to $S$ from farthest to smallest.

    \STATE Let $\pv$ be the point that is in the $8\log n$th position in the ordering.

    \STATE Return $\pv$.

    \end{algorithmic}
\end{algorithm}


The algorithm $\iters$ maintains a set of sampled points $S$ and a set of
points $R$ that contains the set of points that are not well represented by the
current sample. The algorithm repeately adds new points to the sample.   By
adding more points to the sample, $S$ will represent more points well.  More
points are added to $S$ until the number of remaining points decreases below
the threshold given in line 2. The point $v$ chosen in line 5 serves as the
pivot to determine which points are well represented: if a point $x$ is closer
to the sample $S$ than the pivot $v$, the point $x$ is considered to be well
represented by $S$ and dropped from $R$.  Finally, $\iters$ returns the union
of $S$ and $R$. Note that $R$ must be returned since $R$ is not well
represented by $S$ even at the end of the while loop.

\subsection{MapReduce Algorithms}

First we show a MapReduce version of \iters and then we give MapReduce
algorithms for the $k$-center and $k$-median problems.   For these algorithms
we assume that for any set $S$ and parameter $\eta$, the set $S$ can be
arbitrarily partitioned into sets of size $|S|/\eta$ by the mappers.  To see
that this is the case, we refer the reader to \cite{KarloffSV10}.

\begin{algorithm}[h!] \caption{$\itersmr(V, E, k, \eps)$:} \label{alg:itersmr}
    \begin{algorithmic}[1]

    \STATE Set $S \leftarrow \emptyset$, $R \leftarrow V$.

    \WHILE{ $|R| > \frac{4}{\eps} k n^{\eps} \log n$}

    \STATE The mappers arbitrarily partition $R$ into $\lceil |R|/n^\eps \rceil
    $ sets of size at most $\lceil n^\eps \rceil$.  Each of these sets is
    mapped to a unique reducer.

    \STATE For a reducer $i$, let $R^i$ denote the points assigned to the
    reducer.   Reducer $i$ adds each point in $R^i$ to a set $S^i$
    independently with probability $\frac{9k n^\eps}{|R|} \log n$ and also adds
    each point in $R^i$ to a set $H^i$ independently with probability  $\frac{4
    n^\eps}{|R|} \log n $.

    \STATE Let $H := \bigcup_{1 \leq i \leq \lceil n^\eps \rceil} H^i$ and $S
    := S \cup (\bigcup_{1 \leq i \leq \lceil n^\eps \rceil} S^i)$.  The mappers
    assign $H$ and $S$ to a single machine along with all edge distances from
    each point in $H$ to each point in $S$.

    \STATE The reducer whose input is $H$ and $S$ sets $\pv \leftarrow
    \select(H,S)$.

   \STATE \label{line:largemem} The mappers arbitrarily partition the points in
   $R$ into $\lceil n^{1-\eps}\rceil $ subsets, each of size at most $\lceil
   |R|/n^{1-\eps}\rceil$. Let $R^i$ for $1 \leq i \leq \lceil n^{1-\eps}\rceil$
   denote these sets. Let $\pv$, $R^i$, $S$, the distances between each point
   in $R^i$ and each point in $S$ be assigned to reducer $i$.

   \STATE Reducer $i$ finds the distance of each point $x \in R^i$ to $S$.  The
   point $x$ is removed from $R^i$ if this distance is smaller than the
   distance of $\pv$ to $S$.

    \STATE   Let $R := \bigcup_{i \in [\eta]} R^i$.

    \ENDWHILE

    \STATE Output $ C := S \cup R$

    \end{algorithmic}
\end{algorithm}

The following propositions give the theoretical guarantees of the algorithm;
these propositions can also serve as a guide for choosing an appropriate value
for the parameter $\eps$.  If the probability of an event is $1 - O(1/n)$, we
say that the event occurs with high probability, which we abbreviate as w.h.p.
The first two propositions follow from the fact that, w.h.p., each iteration of
$\iters$ decreases the number of remaining points --- i.e., the size of the set
$R$ --- by a factor of $\Theta(n^\eps)$.  We give the proofs of these
propositions in the next section.  Note that the propositions imply that our
algorithm belongs to $\mrc^0$.

\begin{proposition}\label{prop:num_rounds}
    The number of iterations of the while loop of
    \iters is at most $O(\frac{1}{\eps})$ w.h.p.
\end{proposition}

\begin{proposition}\label{prop:final_sample_size}
    The set returned by \iters  has size
    $O(\frac{1}{\eps} kn^\eps \log n)$ w.h.p.
\end{proposition}
\begin{proposition}\label{prop:sc-mapreduce}
     \itersmr is a MapReduce algorithm that requires  $O ( \frac{1}{\eps}$)
    rounds when machines have memory  $O(k n^\delta)$ for a constant
    $\delta > 2\eps$    w.h.p.
\end{proposition}
\begin{proof}
    Consider a single iteration of the while loop.  Each iteration
    takes a constant number of MapReduce rounds. By
    Proposition~\ref{prop:num_rounds}, the number of iterations of
    this loop is $O(\frac{1}{\eps})$, and therefore the number of
    rounds is $O(\frac{1}{\eps})$.
    The memory needed on a machine is dominated by the memory required
    by Step (\ref{line:largemem}).  The size of $S$ is
    $O(\frac{1}{\eps} kn^\eps \log n)$ by
    Proposition~\ref{prop:final_sample_size}.  Further, the size of
    $R^i$ is at most $n/n^{1-\eps} = n^{\eps}$.  Let $\eta$ be the
    maximum number of bits needed to represent the distance from one
    point to another. Thus the total memory needed on a machine is
    $O(\frac{1}{\eps} kn^\eps \log n \cdot n^{\eps} \cdot \eta)$, the
    memory needed to store the distances from points in $R^i$ to
    points in $S$.  By assumption $\eta = O(\log n)$, thus
    the total memory needed on a machine is upper bounded by
    $O(\frac{1}{\eps} kn^{2\eps} \log^2 n )$.  By setting $\delta$ to
    be a constant slightly larger than $2\eps$, the proposition
    follows.
\end{proof}

Once we have this sampling algorithm, our algorithm $\mapkc$ for the $k$-center
problem is fairly straightforward.  This is the algorithm considered in
Theorem~\ref{thm:kcenter}.  The memory needed by the algorithm is dominated by
storing the pairwise distances between points in $C$ on a single machine.  By
Proposition~\ref{prop:final_sample_size} and the assumption that the maximum
distance between any two points can be represented using $O(\log n)$ bits,
w.h.p. the memory needed is $O((\frac{1}{\eps} kn^\eps \log n)^2 \cdot \log n)
= O(n^\delta k^2)$, where $\delta$ is a constant greater than $2\eps$.

\begin{algorithm}[h!] \caption{$\mapkc(V, E, k, \eps)$:} \label{alg:mapkc}
    \begin{algorithmic}[1]
    \STATE Let $C \leftarrow \iters(V, E, k, \eps)$.

    \STATE Map $C$ and all of pairwise distances between points in $C$
    to a reducer.

    \STATE The reducer runs a $k$-center clustering algorithm
    $\script{A}$ on $C$.

    \STATE Return the set constructed by $\script{A}$.
    \end{algorithmic}
\end{algorithm}

However, for the $k$-median problem, the sample must contain more
information than just the set of sampled points.  This is because
the $k$-median objective considers the sum of the distances to the
centers. To ensure that we can map a good solution for the points
in the sample to a good solution for all of the points, for each
unsampled point $x$, we select the sampled point that is closest to
$x$ (if there are several points that are closest to $x$, we pick
one arbitrarily). Additionally, we assign a weight to each sampled
point $y$ that is equal to the number of unsampled points that
picked $y$ as its closest point. This is done so that, when we
cluster the sampled points on a single machine, we can take into
account the effect of the unsampled points on the objective. For a
point $x$ and a set of points $A$, let $d(x, A)$ denote the minimum
distance from the point $x \in V$ to a point in $A$, i.e., $d(x, A)
= \min_{y \in A} d(x,y)$. The algorithm \mapkm is the following.

\begin{algorithm}[h!] \caption{$\mapkm(V, E, k, \eps)$:} \label{alg:mapkc}
   \begin{algorithmic}[1]
    \STATE Let $C \leftarrow \itersmr(V, E, k, \eps)$

    \STATE The mappers arbitrarily partition $V$ into $\lceil
    n^{1-\eps} \rceil$ sets of size at most $\lceil n^{\eps} \rceil$.
    Let $V^i$ for $1 \leq i \leq \lceil n^{1-\eps} \rceil$ be the
    partitioning.

   \STATE The mappers assign $V^i$, $C$ and all distances between
   points in $V^i$ and $C$ to reducer $i$ for all $1 \leq i \leq
   \lceil n^{1-\eps} \rceil$.

    \STATE Each reducer $i$, for each $y \in C$, computes $w^i(y) =
    |\{x \in V^i \setminus C \;|\; d(x, y) = d(x, C)\}|$.

   \STATE Map all of the weights $w^i(\cdot)$, $C$ and the pairwise
   distances between all points in $C$ to a single reducer.

    \STATE The reducer computes $w(y) = \sum_{i \in [m]} w^i(y) +1$
    for all $y \in C$.

    \STATE The reducer runs a weighted $k$-median clustering algorithm
    $\script{A}$ on that machine with $\left<C, w, k\right>$ as input.

    \STATE Return the set constructed by $\script{A}$.
   \end{algorithmic}
\end{algorithm}

The \mapkm algorithm performs additional rounds to give a weight to
each point in the sample $C$. We remark that these additional rounds
can be easily removed by gradually performing this operation in each
iteration of $\itersmr$.  The maximum memory used by a machine in
\mapkm is bounded similarly as \mapkc. The proof of all propositions
and theorems will be given in the next section. The algorithm \mapkm
is the algorithm considered in Theorem~\ref{thm:kmedian}.  Notice that
both \mapkm and $\mapkc$ use some clustering algorithm as a
subroutine.  The running times of these clustering algorithms depend on
the size of the sample and therefore there is a trade-off between the running
times of these algorithms and the number of MapReduce rounds.

\section{Analysis}

\subsection{Subroutine: Iterative-Sample}

This section is devoted to the analysis of \iters, the main subroutine
of our clustering algorithms.
Before we give the analysis, we introduce some notation.  Let $S^*$
denote any set. We will show several lemmas and theorems that hold for
any set $S^*$, and in the final step, we will set $S^*$ to be the
optimal set of centers. The reader may read the lemmas and theorems
assuming that $S^*$ is the optimal set of centers.
We assign each point $x \in V$ to its closest point in $S^*$, breaking
ties arbitrarily but consistently.  Let $x^{S^*}$ be the point in
$S^*$ to which $x$ is assigned; if $x$ is in $S^*$, we assign $x$ to
itself. Let $S^*(x)$ be the set of all points assigned to $x \in S^*$.

We say that a point $x$ is \emph{satisfied} by $S$ with respect to $S^*$ if
$d(S, x^{S^*}) \leq d(x, x^{S^*})$.
If $S$ and $S^*$ are clear from the context, we will simply say that
$x$ is satisfied.  We say that $x$ is \emph{unsatisfied} if it is not
satisfied. Throughout the analysis, for any point $x$ in $V$ and any
subset $S \subseteq V$, we will let $x^S$ denote the point in $S$ that
is closest to $x$.
%

We now explain the intuition behind the definition of ``satisfied".
Our sampling subroutine's output $C$ may not include each center in
$S^*$. However, a point $x$ could be ``satisfied", even though
$x^{S^*} \notin C$, by including a point in $C$ that is closer to
$x$ than $x^{S^*}$. Intuitively, if all points are satisfied, our
sampling algorithm returned a very representative sample of all
points, and our clustering algorithms will perform well. However, we cannot guarantee that all points are
satisfied. Instead, we will show that the number of unsatisfied
points is small and their contribution to the clustering cost is
negligible compared to the satisfied points' contribution. This will
allow us to upper bound the distance between the unsatisfied points
and the final solution constructed by our algorithm by the cost of
the optimal solution. 

Since the sets described in $\iters$ change in each iteration, for the
purpose of the analysis, we let $R_{\ell}$, $S_{\ell}$, and $H_{\ell}$
denote the sets $R$, $S$, and $H$ at the beginning of iteration $\ell$.
Note that $R_1 = V$ and $S_1 = \emptyset$. Let $D_{\ell}$ denote the set of
points that are removed (deleted) during iteration $\ell$. Note that $R_{\ell +
1} = R_{\ell} - D_{\ell}$. Let $U_{\ell}$ denote the set of points in
$R_{\ell}$ that are not satisfied by $S_{\ell + 1}$ with respect to $S^*$. Let
$C$ denote the set of points that $\iters$ returns.  Let $U$ denote the set of
all unsatisfied points by $C$ with respect to $S^*$.  If one point is satisfied
by $S_{\ell}$ with respect to $S^*$ then it is also satisfied by $C$ with
respect to $S^*$, and therefore $U \subseteq \bigcup_{ \ell \geq 1} U_{\ell}$.

We start by upper bounding $|U_{\ell}|$, the number of unsatisfied
points at the end of iteration $\ell$.

\begin{lemma} \label{lem:unsatisfied}
    Let $S^*$ be any set with no more than $k$ points. Consider
    iteration $\ell$ of $\iters$, where $\ell \geq 1$.  Then $\Pr
    \left [ |U_{\ell}| \geq \frac{|R_l|}{3 n^\eps} \right] \leq
    \frac{1}{n^2}$.
\end{lemma}

\begin{proof}
   Consider any point $y$ in $S^*$.  Recall that $S^*(y)$
    denotes the set of all points that are assigned to $y$. 
    Note that it suffices to show that
           $$\Pr \left[ |U_{\ell} \cap S^*(y) \cap R_{\ell}| \geq
        \frac{|R_{\ell}|}{3kn^\eps} \right] \leq 
        \frac{1}{n^3}$$
     This is because the lemma would follow by taking the union bound over all points in
    $S^*$ (recall that $|S^*| \leq k \leq n$). 
    Hence we focus on bounding the probability that 
      the event $|U_{\ell} \cap S^*(y) \cap R_{\ell}| \geq
        \frac{|R_{\ell}|}{3kn^\eps} $  occurs. 
        The event implies that none of the $\frac{|R_{\ell}|}{3kn^\eps} $
        closest points in $S^*(y) \cap R_{\ell}$ from $y$ 
        was added to $S_\ell$. 
        This is because  if any of such points were added to $S_\ell$, then all points 
        in $S^*(y) \cap R_{\ell}$ farther than the point from $y$  would be satisfied. 
        Hence we have
        $$\Pr \left[ |U_{\ell} \cap S^*(y) \cap R_{\ell}| \geq
        \frac{|R_{\ell}|}{3kn^\eps} \right] \leq (1 - \frac{9k n^\eps
        }{|R_{\ell}|} \log n )^{\frac{|R_{\ell}|}{3k n^\eps}} \leq
        \frac{1}{n^3}$$
This completes the proof. 
\end{proof}

Recall that we selected a threshold point $\pv$ to discard the points
that are well represented by the current sample $S$. Let
$\rank_{R_{\ell}}(v)$ denote the number of points $x$ in $R_{\ell}$
such that the distance from $x$ to $S$ is greater than the distance
from $v$ to $S$. The proof of the following lemma follows easily from
the Chernoff inequality.

\begin{lemma}
    Let $S^*$ be any set with no more than $k$ points. Consider any
    $\ell$-th iteration of the while loop of $\iters$.  Let
    $\pv_{\ell}$ denote the threshold in the current iteration, i.e. the
    $(8 \log n)$-th farthest point in $H_{\ell}$ from $S_{\ell + 1}$.
    Then we have $\Pr[ \frac{|R_{\ell}|}{n^\eps} \leq
    \rank(\pv_{\ell}) \leq \frac{4|R_{\ell}|}{n^\eps}] \geq 1 -
    \frac{2}{n^2}$.
\end{lemma}

\
\begin{proof}
    Let $r = \frac{|R_{\ell}|}{n^\eps}$. Let $N_{\leq r}$ denote the
    number of points in $H_{\ell}$ that have ranks smaller than $r$,
    i.e.  $N_{\leq r} = |\{ x \in H_{\ell} \; | \; \rank_{R_{\ell}}(x)
    \leq r \}|$.  Likewise, $N_{\leq 4r} = |\{ x \in H_{\ell} \; | \;
    \rank_{R_{\ell}}(x) \leq 4r \}|$.  Note that $\Ex [N_{\leq r}] = 4
    \log n$ and $\Ex [N_{\leq 4r} ]= 16 \log n$.  By Chernoff
    inequality, we have $\Pr \left[ N_{\leq r} \geq 8 \log n \right]
    \leq \frac{1}{n^2}$ and $\Pr \left[ N_{\leq 4r} \leq 8 \log n
    \right] \leq \frac{1}{n^2}$.  Hence the lemma follows.
\end{proof}

\begin{corollary}\label{cor:R_decrease}
    Consider any $\ell$-th iteration of the while loop of \iters. Then $\Pr [
    \frac{|R_{\ell}|}{n^\eps} \leq |R_{\ell + 1}| \leq
    \frac{4|R_{\ell}|}{n^\eps}] \geq 1 - \frac{2}{n^2}.$
\end{corollary}

The above corollary immediately implies Proposition~\ref{prop:num_rounds} and
\ref{prop:final_sample_size}.  Now we show how to map each unsatisfied point to
a satisfied point such that no two unsatisfied points are mapped to the same
satisfied point; that is, the map is injective.  Such a mapping will allow us
to bound the cost of unsatisfied points by the cost of the optimal solution.
The following theorem is the core of our analysis. The theorem defines a
mapping $p: U \rightarrow V$; for each point $x$, we refer to $p(x)$ as the
\emph{proxy} point of $x$.

\newcommand{\Uf}{U_{final}}

\begin{theorem} \label{thm:approx-dist}
    Consider any set $S^* \subseteq V$.  Let $C$ be the set of points
    returned by $\iters$.  Let $U$ be the set of all points in $V - C$
    that are unsatisfied by $C$ with respect to $S^*$. Then w.h.p., there
    exists an injective function $p: U \rightarrow V \setminus U$ such
    that, for any $x \in U$, $d(p(x), S^*) \geq d(x, C)$.
\end{theorem}

\begin{proof}
    Throughout the analysis, we assume that $|U_{\ell}| \leq |R_{\ell}| /
    (3n^\eps)$ and $\frac{|R_{\ell}|}{n^\eps} \leq |R_{\ell+1}| \leq
    \frac{4|R_{\ell}|}{n^\eps}$ for each iteration $\ell$.  By
    Lemma~\ref{lem:unsatisfied}, Corollary~\ref{cor:R_decrease}, and a
    simple union bound, it occurs w.h.p.

    Let $\ell_f$ denote the final iteration.  Let $A(\ell) := R_{\ell
    + 1} \setminus U_{\ell}$.  Roughly speaking, $A(\ell)$ is a set of
    candidate points, to which each $x \in U_{\ell} \cap D_{\ell}$ is
    mapped. Formally, we show the following properties:
    \begin{enumerate}
        \item for any $x \in U_{\ell} \cap D_{\ell}$ and any $y \in
        A(\ell)$, $d(x, S_{\ell}) \leq d(y, S_{\ell}) \leq d(y, S^*)$.
        \item $|A(\ell)| \geq \sum_{\ell' \geq \ell}^{1/\eps}
        |U_{\ell}|$.
        \item $\bigcup_{\ell =1}^{\ell_f}(U_{\ell} \cap D_{\ell})
        \supseteq U$.
    \end{enumerate}
    The first property holds because any point in $R_{\ell+1} =
    R_{\ell} - D_{\ell} \supseteq A(\ell)$ is farther from $S_{\ell +
    1}$ than any point in $D_{\ell}$, by the definition of the
    algorithm.

    The inequality $d(y, S_{\ell}) \leq d(y, S^*)$ is immediate since
    $y$ is satisfied by $S_{\ell}$ for $C^*$.  The second property
    follows since $|A(\ell)| \geq |R_{\ell + 1}| - |U_{\ell}| \geq
    \frac{|R_{\ell}|}{n^\eps} - \frac{|R_{\ell}|}{3 n^\eps} \geq
    \sum_{\ell' \geq \ell}^{1/\eps} |U_{\ell}|$.  The last inequality
    holds because $|U_{\ell}| \leq \frac{|R_{\ell}|}{3 n^\eps}$ and
    $|U_{\ell}|$ decrease by a factor of more than two as $\ell$
    grows.  We now prove the third property. The entire set of points
    $V$ is partitioned into disjoints sets $D_1, D_2, ...,
    D_{{\ell}_f}$ and $R_{{\ell}_f + 1}$. Further, for any $1 \leq
    \ell \leq {\ell}_f$, any point in $U \cap D_{\ell}$ is unsatisfied
    by $S_{\ell}$ with respect to $S^*$, thus the point is also in $U_{\ell}
    \cap D_{\ell}$.  Finally, the set $R_{ {\ell}_f + 1} \subseteq C$
    are clearly satisfied by $C$.


    .

    We now construct $p(\cdot)$ as follows. Starting with $\ell =
    \ell_f$ down to $\ell = 1$, we map each unsatisfied point in
    $U_{\ell} \cap D_{\ell}$ to $|A(\ell)|$ so that no point in
    $A(\ell)$ is used more than once.  This can be done using the
    second property. The requirement of $p(\cdot)$ that for any $x \in
    U$, $d(p(x), S^*) \geq d(x, C)$ is guaranteed by the first
    property.  Finally, we simply ignore the points in $\bigcup_{\ell
    =1}^{{\ell}_f}(U_{\ell} \cap D_{\ell}) \setminus U$.  This
    completes the proof.
\end{proof}

\subsection{MapReduce-KCenter}

This section is devoted to proving Theorem~\ref{thm:kcenter}.  For the
sake of analysis, we will consider the following variant of the
$k$-center problem. In the $\texttt{kCenter}(V, T)$ problem, we are
given two sets $V$ and $T \subseteq V$ of points in a metric space,
and we want to select a subset $S^* \subseteq T$ such that $|S^*| \leq
k$ and $S^*$ minimizes $\max_{x \in V} d(x, S)$ among all sets $S
\subseteq T$ with at most $k$ points. For notational simplicity, we
let $\opt(V, T)$ denote the optimal solution for the problem
$\texttt{kCenter}(V, T)$.  Depending on the context, $\opt(V, T)$ may
denote the cost of the solution.  Since we are interested eventually
in $\opt(V, V)$, we let $\opt := \opt(V,V)$.

\begin{proposition} \label{prop:kcenter-unsat}
    Let $C$ be the set of centers returned by $\iters$. Then w.h.p.
    we have that for any $x \in V$, $d(x, C) \leq  2 \opt$.
\end{proposition}

\begin{proof}
    Let $S^* := \opt$ denote a fixed optimal solution for
    $\texttt{kCenter}(V,V)$.  Let $U$ be the set of all points that
    are not satisfied by $C$ with respect to $S^*$.  Consider any point $x$
    that is satisfied by $C$ concerning $S^*$. Since it is satisfied,
    there exists a point $a \in C$ such that $d(a, x^{S^*}) \leq d(x,
    x^{S^*}) = d(x, S^*)$.  Then by the triangle inequality, we have
    $d(x, C) \leq d(x, a) \leq d(x, x^{S^*}) + d(a, x^{S^*}) \leq
    2d(x, S^*) \leq 2 \max_{y \notin U} d(y, S^*) = 2 \opt$.  Now
    consider any unsatisfied $x$. By Theorem~\ref{thm:approx-dist}, we
    know that w.h.p. there exists a proxy point $p(x)$ for any
    unsatisfied point $x \in U$. Then using the property of proxy
    points, we have
    $d(x, C) \leq d(p(x), S^*) \leq d(p(x), S^*) \leq  \max_{y \notin
    U} d(y, S^*) \leq \opt$.
\end{proof}

\begin{proposition} \label{prop:kcenter-newopt}
    Let $C$ be the set of centers returned by $\iters$. Then w.h.p.
    we have $\opt(C, C) \leq \opt(V, C) \leq \opt.$
\end{proposition}

\begin{proof}
    Since the first inequality is trivial, we focus on proving the
    second inequality.  Let $S^*$ be an optimal solution for
    $\texttt{kCenter}(V, V)$. We construct a set $T \subseteq C$ as
    follows: for each $x \in S^*$, we add to $T$ the point in $C$
    that is closest to $x$. Note that $|T| \leq k$ by construction.
    For any $x \in V$, we have
    \begin{eqnarray*}
        d(x, T) &\leq& d(x, x^{S^*}) + d(x^{S^*}, T)= d(x, x^{S^*}) +
        d(x^{S^*}, x^C) \qquad\qquad\\
        &&\;\;\; \mbox{[Since the closest point in $C$ to $x^{S^*}$ is
        in $T$]}\\
        &\leq& d(x, x^{S^*}) + d(x, x^{S^*}) + d(x, x^C)\\
        &=& 2d(x, S^*) + d(x, C)   \leq 2 \opt + d(x, C)
    \end{eqnarray*}
    By Proposition~\ref{prop:kcenter-unsat}, we know that w.h.p. for
    all $x \in V$, $d(x, C) \leq 2 \opt(V, V)$. Therefore, for all $x
    \in V$, $d(x, T) \leq 4\opt$. Since $\opt (V, C) \leq \opt (V,
    T)$, the second inequality follows.
\end{proof}

\begin{theorem}
    If $\script{A}$ is an algorithm that achieves an
    $\alpha$-approximation for the $k$ center problem, then w.h.p. the
    algorithm $\mapkc$ achieves a $(4\alpha + 2)$-approximation for
    the $k$ center problem.
\end{theorem}

\begin{proof}
    By Proposition~\ref{prop:kcenter-newopt}, $\opt(C,C) \leq 4\opt$.
    Let $S$ be the set returned by $\mapkc$. Since $\script{A}$
    achieves an $\alpha$-approximation for the $k$ center problem, it
    follows that
        $$\max_{x \in C} d(x, S) \leq \alpha \opt(C,C) \leq
        4\alpha\opt$$
    Let $x$ be any point. By Proposition~\ref{prop:kcenter-unsat},
        $$d(x, C) \leq 2\opt$$
    Therefore
        $$d(x, S) \leq d(x^C, S) + d(x, x^C) \leq (4\alpha +
        2)\opt$$
\end{proof}

\noindent
By setting the algorithm $\mathcal{A}$ to be the $2$-approximation of
\cite{DyerF85,Gonzalez85}, we complete the proof
Theorem~\ref{thm:kcenter}.

\subsection{MapReduce-KMedian}
\label{sec:kmedian}
In the following, we will consider the following variants of the
$k$-median problem similar to the variant of the $k$-center problem
considered in the previous section. In the $\texttt{kMedian}(V, T)$
problem, we are given two sets $V$ and $T \subseteq V$ of points in a
metric space, and we want to select a subset $S^* \subseteq T$ such
that $|S^*| \leq k$ and $S^*$ minimizes $\sum_{x \in V} d(x, S)$ among
all sets $S \subseteq T$ with at most $k$ points. We let $\opt(V, T)$
denote a fixed optimal solution for $\texttt{kMedian}(V, T)$ or the
optimal cost depending on the context. Note that we are interested in
obtaining a solution that is comparable to $\opt(V, V)$. Hence, for
notational simplicity, we let $\opt := \opt(V, V)$. In the
$\texttt{Weighted-kMedian}(V, w)$ problem, we are given a set $V$ of
points in a metric space such that each point $x$ has a weight $w(x)$,
and we want to select a subset $S^* \subseteq V$ such that $|S^*| \leq
k$ and $S^*$ minimizes $\sum_{x \in V} w(x) d(x, S)$ among all sets $S
\subseteq V$ with at most $k$ points. Let $\opt^w(V, w)$ denote a
fixed optimal solution for a $\texttt{Weighted-kMedian}(V, w)$.

Recall that $\mapkm$ computes an approximate $k$-medians on $C$ with
each point $x$ in $C$ having a weight $w(x)$. Hence we first show that
we can obtain a good approximate $k$-medians using only the points in
$C$.

\begin{proposition} \label{prop:kmedian-unsat}
    Let $S^*: = \opt$. Let $C$ be the set of centers returned by
    $\iters$. Then w.h.p., we have that $\sum_{x \in V} d(x, C) \leq 3
    \opt.$
\end{proposition}

\begin{proof}
    Let $U$ denote the set of points that are not unsatisfied by $C$
    with respect to $S^*$. By Theorem~\ref{thm:approx-dist}, w.h.p. there
    exist proxy points $p(x)$ for all unsatisfied points. First
    consider any satisfied point $x \notin U$. It follows that there
    exists a point $a \in C$ such that
        $d(a, x^{S^*}) \leq d(x, x^{S^*}) = d(x, S^*)$.
    By the triangle inequality,
        $d(x, C) \leq d(x, a) \leq d(x, x^{S^*}) + d(a, x^{S^*}) \leq 2d(x,
        S^*)$. Hence $\sum_{x \notin U} d(x, C) \leq 2 \opt.$
    We now argue with the unsatisfied points.
        $\sum_{x \in U} d(x, C)  \leq \sum_{x \in U} d(p(x), S^*) \leq \opt.$
    The last inequality is due to property that $p(\cdot)$ is
    injective.
\end{proof}

\begin{proposition} \label{prop:kmedian-newopt}
    Let $C$ be the set returned by $\iters$. Then w.h.p., $\opt(V, C)
    \leq 5\opt$.
\end{proposition}

\begin{proof}
    Let $S^*$ be an optimal solution for $\texttt{kMedian}(V, V)$. We
    construct a set $T \subseteq C$ as follows: for each $x \in S^*$,
    we add to $T$ the point in $C$ that is closest to $x$. By
    construction $|T| \leq k$.  For any $x$, we have
    \begin{eqnarray*}
        d(x, T) &\leq& d(x, x^{S^*}) + d(x^{S^*}, T)\leq d(x, x^{S^*})
        + d(x^{S^*}, x^C) \qquad\qquad \\
        && \;\;\;  \mbox{[The closest point in $C$ to $x^{S^*}$ is in
        $T$]}\\
        &\leq& d(x, x^{S^*}) + d(x, x^{S^*}) + d(x, x^C) = 2d(x, S^*)
        + d(x, C)
    \end{eqnarray*}
    By applying Proposition~\ref{prop:kmedian-unsat}, w.h.p. we have
        $\sum_{x \in V} d(x, T) \leq 2 \sum_{x \in V} d(x, S^*) +
        \sum_{x \in V} d(x, C) \leq 5 \opt.$
    Since $T$ is a feasible solution for $\texttt{kMedian}(V, C)$, it
    follows that $\opt(V, C) \leq 5\opt$.
\end{proof}

So far we have shown that we can obtain a good approximate solution
for the $\texttt{kMedian}(V, V)$ even when we are restricted to $C$.
However, we need a stronger argument, since $\mapkm$ only sees the
weighted points in $C$ and not the entire point set $V$.

\begin{proposition} \label{prop:kmedian-weighted-opt}
    Consider any subset of points $C \subseteq V$. For each point $y
    \in C$, let $w(y) = |\{x \in V - C \;|\; x^C = y\}| + 1$. Then we
    have $\opt^w(C, w) \leq 2\opt(V,C)$.
\end{proposition}

\begin{proof}
    \newcommand{\Co}{\overline{C}}
    Let $T^* := \opt(V, C)$. Let $\Co := V \setminus C$.  For each
    point $x \in \Co$, we have $d(x, x^C) + d(x, x^{T^*}) \geq d(x^C,
    x^{T^*}) \geq d(x^C, T^*).$ Therefore $\sum_{x \in \overline{C}}
    d(x, T^*) \geq \sum_{x \in \overline{C}} (d(x^C, T^*) - d(x,
    x^C))$. Further we have,
    \begin{eqnarray*}
        2\sum_{x \in \overline{C}} d(x, T^*) &\geq& \sum_{x \in
        \overline{C}} d(x^C, T^*) + \sum_{x \in \overline{C}} (d(x,
        T^*) - d(x, x^C))\\
        &\geq& \sum_{x \in \overline{C}} d(x^C, T^*)\;\; \mbox{[$d(x,
        T^*) \geq d(x, C)$, since $T^* \subseteq C$]}\\
        &=& \sum_{y \in C} \sum_{x \in \overline{C}: x^C = y} d(y,
        T^*) = \sum_{y \in C} (w(y) - 1) d(y, T^*)
    \end{eqnarray*}
    Hence we have $\opt(V, C) = 2\sum_{x \in V} d(x, T^*) \geq \sum_{y
    \in C} w(y) d(y, T^*)$. Since $T^*$ is a feasible solution for
    $\texttt{Weighted-kMedian}(C, w)$, it follows that $\opt^w(C,w)
    \leq 2\opt(V, C)$.
\end{proof}

\begin{theorem}
    If $\script{A}$ is an algorithm that achieves an
    $\alpha$-approximation for $\texttt{Weighted-kMedian}$, w.h.p. the
    algorithm $\mapkm$ achieves a $(10\alpha + 3)$-approximation for
    $\texttt{kMedian}$.
\end{theorem}

\begin{proof}
    It follows from Proposition~\ref{prop:kmedian-newopt} and
    Proposition~\ref{prop:kmedian-weighted-opt} that w.h.p.,$\opt^w
    (C, w) \leq 10\opt$.

    Let $S$ be the set returned by $\mapkm$. Since $\script{A}$
    achieves an $\alpha$-approximation for
    $\texttt{Weighted-kMedian}$, it follows that
        $$\sum_{y \in C} w(y) d(y, S) \leq \alpha \opt^w(C, w) \leq
        10\alpha \opt$$
    We have
    \begin{eqnarray*}
        \sum_{x \in V} d(x, S) &=& \sum_{y \in C} d(y, S) +
        \sum_{x \in \overline{C}} d(x, S)\\
        &\leq& \sum_{y \in C} d(y, S) +
        \sum_{x \in \overline{C}} d(x, (x^C)^S)\\
        &\leq& \sum_{y \in C} d(y, S) + \sum_{x \in \overline{C}}
        (d(x, x^C) + d(x^C, S))\\
        &=& \sum_{y \in C} w(y) d(y, S) + \sum_{x \in \overline{C}}
        d(x, C)
    \end{eqnarray*}
    By Proposition~\ref{prop:kmedian-unsat} (with $S^*$ equal to
    $\opt(V,V)$), we get that
        $$\sum_{x \in V} d(x, C) \leq 3\sum_{x \in V} d(x, S^*) =
        3\opt$$
    Therefore
        $$\sum_{x \in V} d(x, S) \leq (10 \alpha + 3) \opt$$
\end{proof}

\noindent
Recall that there is a $(3 + 2/c)$ approximation algorithm for
$k$-median that runs in $O(n^c)$ time \cite{AryaGKMMP04, GuptaT08}.
In order to complete the proof of Theorem~\ref{thm:kmedian}, we pick a
constant $c$ and we use the $(3 + 2/c)$-approximation algorithm.

\section{Experiments}\label{sec:exp}

\newcommand{\local}{\texttt{LocalSearch}\xspace}
\newcommand{\parlloyd}{\texttt{Parallel-Lloyd}\xspace}
\newcommand{\partitionlloyd}{\texttt{Divide-Lloyd}\xspace}
\newcommand{\partitionlocal}{\texttt{Divide-LocalSearch}\xspace}
\newcommand{\samlocal}{\texttt{Sampling-LocalSearch}\xspace}
\newcommand{\samlloyd}{\texttt{Sampling-Lloyd}\xspace}
\newcommand{\std}{\sigma}

\begin{figure*}
    \begin{center}
\begin{tabular}{|l|l|r|r|r|r|r|r|r|}
 \hline
         & Number of points  & 10,000 & 20,000 & 40,000  &100,000 & 200,000 & 400,000 & 1,000,000  \\
 \hline
  cost  &  \parlloyd & 1.000  & 1.000& 1.000 & 1.000 & 1.000& 1.000& 1.000     \\

       &  \partitionlloyd & 1.030  & 1.088 & 1.132 & 1.033 &1.051 & 0.994& 1.023   \\

        &   \partitionlocal &0.999  & 1.024 & 1.006 & 0.999 & 1.008& 0.999& 1.010      \\

       & \samlloyd & 1.086 & 1.165 & 1.051& 1.138 & 1.068 & 1.095 &  1.132   \\

        & \samlocal & 1.018    & 1.019 &1.011 & 1.006& 1.024& 1.025  & 1.029   \\

        & \local  &0.948     & 0.964 &0.958 & N/A & N/A& N/A& N/A   \\
\hline
  time  &  \parlloyd   & 0.0 & 3.3 & 6.0 & 18.0 &29.3 &52.7 & 205.7   \\

       &  \partitionlloyd  &0.0 & 0.3&0.3 & 1.3&1.0 &1.0 &  2.7  \\

        &   \partitionlocal   & 5.0 &5.0 &5.0 &6.0 &9.0 &19.0 & 70.7    \\

       & \samlloyd &0.3 & 0.0 & 0.3 & 0.7 &1.3 &3.0 &  4.0 \\

        & \samlocal   & 1.7& 2.3& 3.0& 4.3&6.0  &  8.3 & 11.0\\

        & \local      & 666.7 & 943.0& 2569.3& N/A & N/A & N/A & N/A  \\  \hline
\end{tabular}
  \label{tab:med:num_points:cost}
    \end{center} \vspace{-4mm}
    \small{\caption{The relative cost and running time of clustering algorithms when the number of points is not too large.  The costs are normalized to that of \parlloyd.  The running time is given in seconds.\label{tbl:small}}} \vspace{-1.5mm}
\end{figure*}

\begin{figure*}
    \begin{center}
\begin{tabular}{|l|l|r|r|r|}
 \hline
  & Number of points            & 2,000,000    & 5,000,000 &   10,000,000 \\
 \hline
  cost      & \parlloyd         & 1.000         & 1.000         & 1.000   \\
           & \partitionlloyd   & 1.018         &  1.036         & 1.000   \\

            &  \samlloyd    & 1.064     & 1.106     & 1.073    \\
             & \samlocal         & 1.027     & 1.019     & 1.015   \\

  \hline
  time      & \parlloyd         & 458.0   & 1333.3    & 702.3 \\
            & \partitionlloyd   &   8.3       & 24.7         & 50.7  \\
            &  \samlloyd          &  8.0 &   18.3    &38.0 \\
              & \samlocal         & 16.3     & 29.3     & 86.3   \\

 \hline
\end{tabular}
  \label{tab:med:num_points:cost}
    \end{center} \vspace{-4mm}
    \small{\caption{The relative cost and running time of the scalable algorithms when the number of points are large. The costs were normalized to that of \parlloyd. The running time is given in seconds.\label{tbl:large}}} \vspace{-2.5mm}
\end{figure*}

In this section we give an experimental study of the algorithms
introduced in this paper. The focus for this section is on the
$k$-median objective because this is where our algorithm gives the
largest increase in performance. Unfortunately, our sampling
algorithm does not perform well for the $k$-center metric.   This is
because the $k$-center objective is quite sensitive to sampling.
Since the maximum distance from a point to a center is considered in
the objective, if the sampling algorithm misses even one important
point then the objective can substantially increase. From now on, we
only consider the $k$-median problem. In the following, we describe
the algorithms we tested and we give an overview of the experiments
and the results.

\subsection{Implemented Algorithms}

We compare our algorithm \mapkm to several algorithms. Recall that \mapkm uses
$\iters$ as a sub-procedure and we have shown that \mapkm gives a constant
approximation when the local search algorithm \cite{AryaGKMMP04, GuptaT08} is
applied on the sample that was obtained by $\iters$.  We also consider Lloyd's
algorithm together with the sampling procedure $\iters$; that is, in \mapkm,
the algorithm $\mathcal{A}$ is Lloyd's algorithm and it takes as input the
sample constructed by $\iters$. Note that Lloyd's algorithm does not give an
approximation guarantee. However, it is the most popular algorithm for
clustering in practice and therefore it is worth testing its performance.  We
will use \samlocal to refer to \mapkm with the local search algorithm as
$\mathcal{A}$ and we will use \samlloyd to refer to \mapkm with Lloyd's
algorithm as $\mathcal{A}$. Note that the only difference between \samlocal and
\samlloyd is the clustering algorithm chosen as $\script{A}$ in $\mapkm$.

We also implement the local search algorithm and Lloyd's algorithm
without sampling. The local search algorithm, denoted as $\local$, is
the only sequential algorithm among all algorithms that we implemented
\cite{AryaGKMMP04, GuptaT08}. We implement a parallelized version of
Lloyd's algorithm, \parlloyd \cite{Lloyd82,blog,googlelecture}. This
implementation of Lloyd's algorithm parallelizes a sub-procedure of
the sequential Lloyd's algorithm.  The parallel version of Lloyd's
gives the same solution as the sequential version of Lloyd's; the only
difference between the two implementations is the parallelization.  We
give a more formal description of the parallel Lloyd's algorithm
below.

Finally, we implement clustering algorithms based on a simple partitioning
scheme used to adapt sequential algorithms to the parallel setting. In the
partition scheme \dividekm we consider, points are partitioned into $\ell$ sets
of size $\lceil \frac{n}{\ell} \rceil$. In parallel, centers are computed for
each of the partitions.  Then all of the centers computed are combined into a
single set and the centers are clustered. We formalize this in the algorithm
\dividekm. We evaluated the local search algorithm and Lloyd's algorithm
coupled with this partition scheme. Throughout this section, we use
\partitionlocal for the local search together with this partition scheme. We
call Lloyd's algorithm coupled with the partition scheme as \partitionlloyd. We
give the details of the partition framework \dividekm shortly.

The following is a summary of the algorithms we implemented: \vspace{-1mm}

\begin{itemize}
    \item \local: Local Search \vspace{-1mm}
    \item \parlloyd: Parallel Lloyd's \vspace{-1mm}
    \item \samlocal: Sampling and Local Search \vspace{-1mm}
    \item \samlloyd: Sampling and Lloyd's \vspace{-1mm}
    \item \partitionlocal: Partition and Local Search \vspace{-1mm}
    \item \partitionlloyd: Partition and Lloyd's \vspace{-1mm}
\end{itemize}

A careful reader may note that Lloyd's algorithm is generally used for the
$k$-means objective and not for $k$-median.  Lloyd's algorithm is more commonly
used for $k$-means, but it can be used for $k$-median as well, and it is one of
the most popular clustering algorithms in practice. We note that the
parallelized version of Lloyd's algorithm we introduce only works with points
in Euclidean space.

\vspace{2.5mm}
\noindent \textbf{Parallel Lloyd's Algorithm:}
We give a sketch of parallelized implementation of Lloyd's algorithm used in the
experiments. More details can be found in \cite{blog,googlelecture}.
The algorithm begins by partitioning the points evenly across the
machines and these points will remain on the machines. The algorithm
initializes the $k$ centers to an arbitrary set of points. In each
iteration, the algorithm improves the centers as
follows. The mapper sends the $k$ centers to each of the machines.  On
each machine, the reducer clusters the points on the machine by
assigning each point to its closest center.  For each cluster, the
average\footnote{Recall that the input to Lloyd's algorithm is a set
of points in Euclidean space. The average of the points is the point
in Euclidean space whose coordinates are the average of the
coordinates of the points.} of the points in the cluster is computed
along with the number of points assigned to the center. The mappers
map all this information to a single machine.  For each center, the
mappers aggregate the points assigned to the center over all
partitions along with the centers, and then the reducers update the
center to be the average of these points.  It is important to note
that the solution computed by the algorithm is the same as the
sequential version of Lloyd's algorithm.

\vspace{2.5mm}
\noindent \textbf{Partitioning Based Scheme:}
We describe the partition scheme \dividekm that is used for the
\partitionlocal and \partitionlloyd algorithms. The  algorithm
\dividekm is a partitioning-based parallelization of any arbitrary
sequential clustering algorithm. We note that this algorithm and the
following analysis have also been considered by Guha \etal
\cite{GuhaMMMO03} in the streaming model.

\begin{algorithm}[h!] \caption{$\dividekm(V, E, k, \ell)$:}
    \begin{algorithmic}[1]
    \STATE Let $n = |V|$.

    \STATE The mappers arbitrarily partition $V$ into disjoint sets
    $S_1, \cdots, S_{\ell}$, each of size $\Theta(n/\ell)$.

    \FOR{$i = 1$ to $\ell$}
        \STATE The mapper assigns $S_i$ and all the distances between
        points in $S_i$ to reducer $i$.

        \STATE Reducer $i$ runs a $k$-median clustering algorithm
        $\script{A}$ with $\left<S_i, k \right>$ as input to find a
        set $C_i \subseteq S_i$ of $k$ centers.

        \STATE Reducer $i$ computes, for each $y \in C_i$,
        $w(y) = |\{x \in S^i \setminus C_i \;|\; d(x, y) = d(x,
        C_i)\}| + 1$.
    \ENDFOR
    \STATE Let $C = \bigcup_{i = 1}^{\ell} C_i$.

    \STATE \label{line:mempar}The mapper sends $C$, the pairwise
    distances between points in $C$ and the numbers $w(\cdotp)$ to a
    single reducer.

    \STATE The reducer runs a $\texttt{Weighted-kMedian}$ algorithm
    $\script{A}$ with $\left<C, w, k\right>$ as input.

    \STATE Return the set constructed by $\script{A}$.
    \end{algorithmic}
\end{algorithm}

It is straightforward to verify that setting $\ell = \sqrt{n / k}$
minimizes the maximum memory needed on a machine; in the following, we
assume that $\ell = \sqrt{n / k}$. The total memory used by the
algorithm is $O(k n \log{n})$. (Recall that we assume that the
distance between two points can be represented using $O(\log n)$
bits.) Additionally, the memory needed is also $\Omega(kn)$, since in
Step (\ref{line:mempar}), $\Theta(\sqrt{n / k})$ sets of $k$ points
are sent to a single machine along with their pairwise distances.  The
following proposition follows from the algorithm description.

\begin{proposition}
     \dividekm runs in $O(1)$ MapReduce rounds.
\end{proposition}

From the analysis given in \cite{GuhaMMMO03}, we have the following
theorem which can be used to bound the approximation factor of
$\dividekm$.

\begin{theorem}[Theorem~{2.2} in \cite{GuhaMMMO03}]
\label{thm:divide-main}
    Consider any set of $n$ points arbitrarily partitioned into
    disjoint sets $S_1, \cdots, S_{\ell}$. The sum of the optimum
    solution values for the $k$-median problem on the $\ell$ sets of
    points is at most twice the cost of the optimum $k$-median
    problem solution for all $n$ points, for any $\ell >0$.
\end{theorem}

\begin{corollary}[\cite{GuhaMMMO03}]
    If the algorithm $\script{A}$ achieves an $\alpha$-approximation
    for the $k$-median problem, the algorithm $\dividekm$ achieves a
    $3\alpha$-approximation for the $k$-median problem.
\end{corollary}

By this Corollary, note that \partitionlocal is a constant
factor approximation.

\subsection{Experiment Overview}
We generate a random set of points in $\mathds{R}^3$. Our data set
consists of $k$ centers and randomly generated points around the
centers to create clusters. The $k$ centers are randomly positioned in
a unit cube. The number of points generated within a cluster is
sampled from a Zipf distribution. More precisely, let $\{C_i \}_{1
\leq i \leq k}$ be the set of clusters. Given a fixed number of
points, a unique point is assigned to the cluster $C_i$ with
probability $i^\alpha / \sum_{i = 1}^k i^\alpha$ where $\alpha$ is the
parameter of the Zipf distribution. Notice that when $\alpha = 0$, all
clusters will have almost the same size and, as $\alpha$ grows, the
sizes of the clusters become more non-uniform.  The distance between a
point and its center is sampled from a normal distribution with a fixed
global standard deviation $\std$. Each experiment with the same
parameter set was repeated three times and the average was calculated.
When running the local search or Lloyd's algorithm, the seed centers
were chosen arbitrarily.

All experiments were performed on a single machine. When running MapReduce
algorithms, we simulated each machine used by the algorithm.  For a given
round, we recorded the time it takes for the machine that ran the longest in
the round. Then we summed this time over all the rounds to get the final
running time of the parallel algorithms. In these experiments, the
communication cost was ignored. More precisely, we ignored the time needed to
move data to a different machine.  The specifications of the machine were
Intel(R) Core(TM) i7 CPU 870 @ 2.93GHz and the memory available was 8GB. We
used the standard \texttt{clock}() function to measure the time for each
experiment.  All parallel algorithms were simulated assuming that there are 100
machines. For the algorithm $\mapkm$ the value of $\eps$ was set to $.1$ for
the sampling probability.

\subsection{Results}
Because of the space constraints, we only give a brief summary of our
results. The data can be found in Figures \ref{tbl:small} and
\ref{tbl:large}. For the data in the figures, the number of points is
the only variable, and other parameters are fixed: $\std = 0.1$,
$\alpha = 0$ and $k = 25$.  The cost of the algorithms' objectives is
normalized to the cost of \parlloyd in the figures.
Figure~\ref{tbl:small} summarizes the results of the experiments on data sets
with at most $10^6$ points,
and Figure~\ref{tbl:large} summarizes the results of the experiments on data
sets with at most $10^7$ points.

Our experiments show that \samlloyd and \samlocal achieve a
significant speedup over \parlloyd (about 20x), a speedup of more than
ten times over \partitionlocal and a significant speedup over \local
(over 1000x) as seen in Figure \ref{tbl:small}. The speedup increases
very fast as the number of points increases. Further, this speedup is
achieved with negligible loss in performance; our algorithm's
objective performs close to the \parlloyd and \local when the number
of points is sufficiently large.

Finally, we compare the performance of \samlocal and \samlloyd with the
performance of \partitionlloyd on the largest data sets; the results are
summarized in Figure \ref{tbl:large}. These algorithms were chosen because they
are the most scalable and perform well; as shown in Figure~\ref{tbl:small},
$\local$ is far from scalable.  Although $\partitionlocal$'s running times are
similar to $\parlloyd$'s, we were not able to run additional experiments with
$\partitionlocal$ because it takes a very long time to simulate on a single
machine.  These additional experiments show that, for data sets consisting of
$5 \times 10^6$ points, the running time of \samlocal is slightly larger than
\partitionlloyd's and the clustering cost of \samlocal is similar to the cost
of \partitionlloyd.  The algorithm \samlloyd achieves a speedup of about 25\%
over \partitionlloyd when the number of points is $10^7$.  Overall the
experiments show that, when coupled with Lloyd's algorithm, our sampling
algorithm runs faster than any previously known algorithm that we considered,
and this speedup is achieved at a very small loss in performance. We also ran
experiments with different settings for the parameters $\alpha$, $k$, and
$\std$, and the results were similar; we omit these results from this version
of the paper.


\section{Conclusion}
In this paper we give the first approximation algorithms for the
$k$-center and $k$-median problems that run in a constant number of
MapReduce rounds.  We note that we have preliminary evidence that the
analysis used for the $k$-median problem can be extended to the
$k$-means problem in Euclidean space; for this problem, our analysis
also gives a MapReduce algorithm that runs in a constant number of
rounds and achieves a constant factor approximation.

\vspace{.25cm} \noindent \textbf{Acknowledgments:} We thank Kyle Fox
for his help with the implementation of the algorithms and several
discussions. We thank Chandra Chekuri for helpful comments on an
earlier draft of this paper. We thank Ravi Kumar and Sergei
Vassilvitskii for helpful comments on the problems considered. We
thank Jeff Erickson for his advice on how to improve the
presentation.


\bibliographystyle{abbrv}
\bibliography{clustermapreduce}

\end{document}